\documentclass{article}
\usepackage{graphicx}
\usepackage{natbib}
\usepackage{amssymb,amsmath, amscd, amsthm}
\newtheorem{thm}{Theorem}[section]

\newtheorem{lem}[thm]{Lemma}
\newtheorem{prop}[thm]{Proposition}
\theoremstyle{definition}

\newtheorem*{criterion}{Criterion}
\theoremstyle{remark}

\numberwithin{equation}{section}
\renewcommand{\phi}{\varphi}
\begin{document}

\title{On Einstein Algebras and Relativistic Spacetimes\thanks{This material is based upon work supported by the National Science Foundation under Grant Nos. 1331126 (Rosenstock \& Weatherall) and DGE 1148900 (Barrett).  The authors can be reached at rosensts@uci.edu, thomaswb@princeton.edu, and weatherj@uci.edu.}}
\author{Sarita Rosenstock \\ Thomas William Barrett\\ James Owen Weatherall}%
\date{}
\maketitle
% ----------------------------------------------------------------
\begin{abstract}
In this paper, we examine the relationship between general relativity and the theory of Einstein algebras. We show that according to a formal criterion for theoretical equivalence recently proposed by \cite{halvorson2012, halvorson2015} and \cite{WeatherallTE}, the two are equivalent theories.
\end{abstract}

\section{Introduction}\label{introduction}

For much of the Carter and Reagan administrations, John Earman argued that Leibniz's view on spacetime could be made precise by appealing to structures that he called \emph{Leibniz algebras}.\footnote{See especially \citet{Earman1977}, though he also discussed the proposal elsewhere \citep{Earman1978, Earman1979, Earman1986, Earman1989, EarmanWEST}.}  On this proposal, rather than representing spacetime as a manifold with some fields on it, one instead begins with an algebraic structure representing possible configurations of matter.  One then shows that, if this algebraic structure is suitably defined, one can use it to reconstruct a manifold with appropriate fields.  Thus, one has a sense in which spacetime arises as a representation of underlying facts about (algebraic) relations between states of matter, which is at least strongly suggestive of Leibniz's remarks in the correspondence with Clarke.\footnote{For the purposes of what follows, we remain agnostic about the historical question of whether this is an adequate way of reconstructing Leibniz's views.  In particular, see \citet{Stein1977} for an importantly different perspective on Leibniz's views.  That said, history aside, we do think this proposal is a plausible way of understanding what ``relationism'' could mean---though Earman himself argued that the Leibniz algebra approach is a kind of ``substantivalism'', albeit at a lower level.}

Later, Earman argued that analogous algebraic structures---what \citet{GerochEA} had previously called \emph{Einstein algebras}---might provide the mathematical setting for an appropriately ``relationist'' approach to general relativity \citep{Earman1986, EarmanWEST}.  The motivation for this proposal was the hole argument, as redeployed by \citet{Stachel1989} and \citet{Earman+Norton}.\footnote{For more on the hole argument, see \citet{Pooley} and \citet{NortonSEP} and references therein.  One sees versions of both the hole argument and the Einstein algebra proposal in \citet{Earman1986}.}  It seems Earman hoped (indeed, argued) that a single Einstein algebra would correspond to an equivalence class of relativistic spacetimes related by the isometries arising in the hole argument.  One might then take the Einstein algebra formalism to do away with ``excess structure'' allegedly appearing in the formalism of relativistic spacetimes---that is, the formalism of manifolds and tensor fields---by equivocating between what Earman and Norton described as ``Leibniz equivalent'' spacetimes.\footnote{For an independent argument that it is misleading to think of the standard formalism of relativity theory as having excess structure, see \citet{WeatherallHoleArg} and \citet{WeatherallGauge}.}  This program ran aground, however, when \citet{Rynasiewicz} observed that one can define points in algebraic approaches to topological spaces, and argued that this meant one could recover precisely the isometries used in the hole argument as isomorphisms between Einstein algebras.  Thus it seemed Einstein algebras could not do the work Earman hoped.

We think Rynasiewicz had this essentially right.  But there is more to say about the relationship between Einstein algebras and relativistic spacetimes.\footnote{\citet{bain2003} also takes this topic back up, though with somewhat different concerns.}  There is a sense in which these are equivalent theories, according to a standard of equivalence recently discussed by \citet{halvorson2012, halvorson2015} and \citet{WeatherallTE}.\footnote{See \citet{barrett2015a} for a discussion of the relationship between this sense of equivalence and others; for a discussion of related notions, see \citet[\S 6.4]{andreka2002}.  For an overview of work on this notion of equivalence, see \citet{WeatherallLandry}.}  More precisely we will show that, on natural definitions of the appropriate categories, the category of Einstein algebras is dual to the category of relativistic spacetimes, where the functors realizing the equivalence preserve the empirical structure of the theories.

Our interest in this problem extends beyond dead horse flagellation, for two reasons.  One is that there has been a flurry of recent work on what it means to say that two theories are equivalent.\footnote{For instance, in addition to the papers already cited, see \citet{North}, \citet{Curiel}, \citet{swansonhalvorsonunpublished}, \citet{BarrettStructure}, \citet{coffey2014}, \citet{halvorson2013}, \citet{glymour2013}, \citet{vanfraassen2014}, \citet{andreka2005},  \citet{barrett2014a}, and \citet{Rosenstock+Weatherall1, Rosenstock+Weatherall2}; of course, there are also some classics on the topic, such as \citet{GlymourTETR, glymour1977, glymour1980}, \citet{quine1975}, and \citet{sklar1982}.}  Here we provide a novel example of two theories that are equivalent by one, but not all, of the standards of equivalence on the table.\footnote{For instance, it is hard to see how these theories could be equivalent by Glymour's criterion of equivalence, since it is not clear how to make sense of ``covariant definability'' when one of the theories is not formulated in terms of fields on a manifold.}  That said, we do not intend to argue that this sense of equivalence is the only one that matters, or even the best one; it is, however, a clear and precise criterion of equivalence between theories, and it seems valuable to explore how it rules on examples of interest as a way of better understanding its significance. The present example is a datum for this ongoing debate.

Indeed, in this regard the example is of particular interest. On reason is that it differs from other examples that have been explored by, for instance, \cite{WeatherallTE} and \citet{Rosenstock+Weatherall1}, in that the categories in question are dual, rather than equivalent.  Another is that the apparent differences between the theories in question---one appears to offer a direct characterization of spatiotemporal geometry; the other uses primarily algebraic methods to described relationships between possible states of matter, from which geometry is reconstructed---are of a different character from other examples in the literature, which may be of probative value as we try to evaluate the merits of the criterion of equivalence.  On this last note, it is striking that the sense of equivalence we discuss here is lurking in the background of the argument \citet{Rynasiewicz} made decades ago, in a different setting, long before any discussion of using methods from category theory to explore equivalence of (scientific) theories;\footnote{That is, the relationship between topological spaces and rings of functions that Rynasiewicz relies on in his paper is precisely a well-known categorical duality, so in a sense, we are following through on a promissory note that the same relationship carries over to relativistic spacetimes and Einstein algebras.} as debates concerning theoretical equivalence continue, it seems important that categorical equivalence has already played an important, even dispositive, role in the philosophy of physics literature.

The second reason the point is worth making is that, as Earman has suggested, the Leibniz/Einstein algebraic formalism does seem to capture a relationist intuition about how spacetime structure relates to possible configurations of matter.\footnote{To be fair, it would not be right to say that Earman simply identified the Leibniz/Einstein algebra formalism with relationism.  For instance, although he suggests that Leibniz algebras are a ``natural candidate'' for models of spacetime answering to Leibniz's desiderata \citep[p. 269]{Earman1979}, he ultimately concludes that they are not fully ``relationist'': as he later puts it, Leibniz/Einstein algebras are ``first-degree non-substantivalist'', yet substantivalist on a deeper level \citep{Earman1986}.  But it is not our purpose to argue about what really constitutes ``relationism''.  We take Earman to be suggesting---and we agree---that the Leibniz/Einstein algebra approach offers a clear and compelling way of capturing the idea that spatio-temporal structure may be reconstructed from relations between possible states of matter.  (Indeed, we do not know of any better way of making this idea precise.) If this is not relationism, but rather some (putative) third way, so be it; as far as we are concerned, that does not affect the interest of the results.  We are grateful to an anonymous referee for pushing us on this point.}   In the background of such arguments is the idea that the standard formalism of relativistic spacetimes is somehow ``substantivalist''.  The results presented here might then be taken as a way of making precise the idea that a suitably ``relationist'' theory, once one spells it out in sufficient detail, is equivalent to the ``substantivalist'' theory, properly construed.\footnote{None of this should be taken to endorse the idea that the standard formalism somehow pushes one to ``substantivalism''.  To the contrary, \citet{WeatherallHoleArg}, following \citet{Stachel1993}, argues that to recover some of the claims \citet{Earman+Norton} attribute to the ``(manifold) substantivalist,'' one needs to \emph{add} structure to the standard formalism---effectively by choosing a global labeling system for spacetime points.}  In particular, the sense of equivalence we consider here is often taken as a standard for determining when two mathematical theories have the same structure.\footnote{See \citet{Baez+etal} and \citet{BarrettCounting}.}  And so, the equivalence we discuss suggests that a relationist theory of general relativity, expressed in these terms, has precisely as much structure as the standard theory.

The remainder of the paper will proceed as follows.  We will begin with some background and mathematical preliminaries concerning general relativity and the notion of equivalence we will consider in what follows.  We will then review the theory of smooth algebras and describe their relationship to smooth manifolds.  Finally, we will define Einstein algebras and prove our main result.  We will conclude with some further remarks on the significance of the result.

\section{General relativity and categorical duality}

In general relativity, we represent possible universes using \emph{relativistic spacetimes}, which are Lorentzian manifolds $(M,g)$, where $M$ is a smooth four dimensional manifold, and $g$ is a smooth Lorentzian metric.\footnote{Here and throughout we assume that smooth manifolds are Hausdorff and paracompact.  For more details on relativistic spacetimes, see \citet{MalamentGR} and \citet{Wald}.  One often requires $M$ to be connected; here we do not make that assumption.  Alternatively, one could require $M$ to be connected, and limit attention below to Einstein algebras whose dual space of points is connected, in the weak topology defined in section \ref{algebras}.} An \textit{isometry} between spacetimes $(M, g)$ and $(M, g')$ is a smooth map $\phi:M\rightarrow M'$ such that $\phi^*(g')=g$, where $\phi^*$ is the pullback along $\phi$.\footnote{Note that we do not require isometries to be diffeomorphisms, so these are not necessarily isomorphisms, i.e., they may not be invertible.} Two spacetimes $(M,g)$, $(M',g')$ are isomorphic, for present purposes, if there is an invertible isometry between them, i.e., if there exists a diffeomorphism $\varphi:M\rightarrow M'$ that is also an isometry.  We then say the spacetimes are \emph{isometric}.

The use of category theoretic tools to examine relationships between theories is motivated by a simple observation: The class of models of a physical theory often has the structure of a category.\footnote{See footnote 6 above for references on this claim. } In what follows, we will represent general relativity with the category $\textbf{GR}$, whose objects are relativistic spacetimes $(M, g)$ and whose arrows are isometries between spacetimes.

According to the criterion for theoretical equivalence that we will consider, two theories are equivalent if their categories of models are ``isomorphic'' in an appropriate sense. In order to describe this sense, we need some basic notions from category theory.\footnote{We take for granted the definitions of a category, a covariant functor, and a contravariant functor. The reader is encouraged to consult \cite{cwm}, \cite{borceux1994}, or \citet{Leinster} for details.} Two (covariant) functors $F:\textbf{C}\rightarrow\textbf{D}$ and $G:\textbf{C}\rightarrow\textbf{D}$ are \textit{naturally isomorphic} if there is a family $\eta_c:Fc\rightarrow Gc$ of isomorphisms of $\textbf{D}$ indexed by the objects $c$ of $\textbf{C}$ that satisfies $\eta_{c'}\circ Ff=Gf\circ\eta_{c}$ for every arrow $f:c\rightarrow c'$ in $\textbf{C}$. The family of maps $\eta$ is called a \textit{natural isomorphism} and denoted $\eta:F\Rightarrow G$. The existence of a natural isomorphism between two functors captures a sense in which the functors are themselves ``isomorphic" to one another as maps between categories. Categories $\textbf{C}$ and $\textbf{D}$ are \textit{dual} if there are contravariant functors $F:\textbf{C}\rightarrow\textbf{D}$ and $G:\textbf{D}\rightarrow\textbf{C}$ such that $GF$ is naturally isomorphic to the identity functor $1_\textbf{C}$ and $FG$ is naturally isomorphic to the identity functor $1_\textbf{D}$. Roughly speaking, $F$ and $G$ give a duality, or contravariant equivalence, between two categories if they are contravariant isomorphisms in the category of categories up to isomorphism in the category of functors. One can think of dual categories as ``mirror images'' of one another, in the sense that the two categories differ only in that the directions of their arrows are systematically reversed.

For the purposes of capturing the relationship between general relativity and the theory of Einstein algebras, we will appeal to the following standard of equivalence.
\begin{criterion}
Theories $T_1$ and $T_2$ are equivalent if the category of models of $T_1$ is dual to the category of models of $T_2$.
\end{criterion}
This criterion is almost exactly the same as the one proposed by \cite{halvorson2012, halvorson2015} and \cite{WeatherallTE}. The only difference is that they require the categories of models in question be \textit{equivalent}, rather than dual.  Equivalence differs from duality only in that the two functors realizing an equivalence are covariant, rather than contravariant. When $T_1$ and $T_2$ are equivalent in either sense, there is a way to ``translate'' (or perhaps better, ``transform'') models of $T_1$ into models of $T_2$, and vice versa.  These transformations take objects of one category---models of one theory---to objects of the other in a way that preserves all of the structure of the arrows between objects, including, for instance, the group structure of the automorphisms of each object, the inclusion relations of ``sub-objects'', and so on.  These transformations are guaranteed to be inverses to one another ``up to isomorphism,'' in the sense that if one begins with an object of one category, maps using a functor realizing (half) an equivalence or duality to the corresponding object of the other category, and then maps back with the appropriate corresponding functor, the object one ends up with is isomorphic to the object with which one began.  In the case of the theory of Einstein algebras and general relativity, there is also a precise sense in which they preserve the empirical structure of the theories.

\section{Smooth algebras and smooth manifolds}

In what follows, by \emph{algebra} we mean a commutative, associative algebra with unit over $\mathbb{R}$---i.e., a real vector space endowed with a commutative, associative product and containing a multiplicative identity.\footnote{Our treatment of smooth algebras follows \citet{Nestruev}.  For more on Einstein algebras in particular, see \citet{GerochEA} and \cite{Heller}.}  By \emph{(algebra) homomorphism} we will mean a map that preserves the vector space operations, the product, and the multiplicative identity; a bijective algebra homomorphism is an \emph{(algebra) isomorphism}.

\subsection{Smooth algebras}\label{algebras}

Let $A$ be an algebra. We denote by $|A|$ the collection of homomorphisms from $A$ to $\mathbb{R}$. The elements of $|A|$  are known as the \emph{points} of the algebra $A$; $|A|$  itself is the \emph{dual space of points}.\footnote{In some treatments of related material, including \citet{Rynasiewicz}, ``points'' are reconstructed as maximal ideals of appropriate rings.  We find the present approach more transparent, mostly because it emphasizes the sense in which points are ``dual'' to smooth functions in the same sense of duality that one encounters elsewhere in geometry and algebra.  But it is closely related to the approach \citet{Rynasiewicz} uses.  In particular, if $x$ is an element of $|A|$, then $\ker(x)$ is an ideal, since if $f\in\ker(x)$, then for any $g\in A$, $x(fg)=x(f)x(g)=0$; moreover, it is maximal, since by linearity, $x$ is surjective, and thus $A/\ker(x)=\mathbb{R}$ and it is well known that for a commutative, unital ring (or algebra) $A$, an ideal $I$ is maximal iff $A/I$ is a field.  Conversely, as we note above, if $A$ is geometric, then $A$ is canonically isomorphic to the space $\tilde{A}$, the maximal ideals of which consist in all functions vanishing at a given point $x\in|A|$.  So points in the sense that Ryansiewicz considers uniquely determine points in the present sense, and vice versa.}  (Note, however, that we do not put any algebraic structure on $|A|$.)  An algebra $A$ is \emph{geometric} if there are no non-zero elements $f\in A$ that lie in the kernel of all of the elements of $|A|$, i.e., if $\bigcap_{p\in |A|} \ker(p) = \{\mathbf{0}\}$.\footnote{Note that the expression ``geometric algebra'' is also used (somewhat more often) to described so-called Clifford algebras.  See, for instance, \citet{Hestenes+Sobczyk} or \citet{Doran+Lasenby}.  The present sense of the term is unrelated.}

We also define the space $\tilde{A}$ as follows:
\[
\tilde{A}=\{\tilde{f}: |A| \rightarrow \mathbb{R}: \exists f\in A \text{ s.t. } \tilde{f}(x) = x(f))\}.
\]
There is a natural algebraic structure on $\tilde{A}$, with operations given by:
\begin{align*}
(\tilde{f}+\alpha \tilde{g})(x) &= \tilde{f}(x) + \alpha \tilde{g}(x)=x(f) + \alpha x(g)\\
(\tilde{f}\cdot\tilde{g})(x) &= \tilde{f}(x)\cdot\tilde{g}(x) = x(f)\cdot g(f)
\end{align*}
There is a canonical map $\tau: A \rightarrow \tilde{A}$ defined by $f \mapsto \tilde{f}$. In general, $\tau$ is a surjective homomorphism.  For geometric algebras, however, $\tau$ is also injective, and thus an isomorphism. We will therefore freely identify a geometric algebra $A$ with $\tilde{A}$ through implicit appeal to $\tau$.

Given a geometric algebra $A$, the \emph{weak topology} on $|A|$ is the coarsest topology on $|A|$ relative to which every element of $A$ (or really, $\tilde{A}$) is continuous.  This defines a Hausdorff topology on $A$.  Now suppose we have an algebra homomorphism $\psi:A\rightarrow B$. Then $\psi$ determines a map $|\psi|: |B| \rightarrow |A|$ defined by $|\psi| :x\mapsto x\circ\psi$.  Any map $|\psi|$ that arises this way is continuous in the weak topology; if $\psi$ is an isomorphism, then $|\psi|$ is a homeomorphism.

Now let $A$ be a geometric algebra, and suppose that $S\subseteq  |A|$ is any subset of its dual space of points.  Then the \emph{restriction} $A_{|S}$ of $A$ to $S$ is the set of all functions $f:S\rightarrow\mathbb{R}$ such that for any point $x\in S$, there exists an open neighborhood $O\subseteq  |A|$ containing $x$, and an element $\bar{f}\in A$ such that $f$ and $\bar{f}$ agree on all points in $O$. One easily verifies that $A_{|S}$ is an algebra, though it is not in general a subalgebra of $A$.

Given any $S\subseteq |A|$, we can define a homomorphism $\rho_S:A\rightarrow A_{|S}$, defined by $f\mapsto f_{|S}$, where here the restriction $f_{|S}$ is meant in the ordinary sense.  The map $\rho_{S}$ is known as the \emph{restriction homomorphism}.  A special case of restriction is restriction to $|A|$, i.e., to the dual space of the algebra, $A_{||A|}$.  This is the collection of all maps on $A$ that are ``locally equivalent'' to elements of $A$.  We will say that $A$ is \emph{complete} if it contains all maps of this form---i.e., if the restriction homomorphism $\rho_{A}:A\rightarrow A_{||A|}$ is surjective.

A complete, geometric algebra $A$ is called \emph{smooth} if there exists a finite or countable open covering $\{U_k\}$ of the dual space $|A|$ such that all the algebras $A_{|U_k}$ are isomorphic to the algebra $C^{\infty}(\mathbb{R}^n)$ of smooth functions on $\mathbb{R}^n$, for some fixed $n$.  Here $n$ is known as the \emph{dimension} of the algebra.  Note that this sense of dimension is unrelated to the dimension of the vector space underlying $A$.

\subsection{The duality of smooth algebras and manifolds} \label{algman}

Smooth algebras and smooth manifolds bear a close relationship to one another. Following \cite{Nestruev}, we present here one standard way to make this relationship precise. We begin by defining two categories: the category $\textbf{SmoothMan}$, whose objects are smooth manifolds and whose arrows are smooth maps, and the category $\textbf{SmoothAlg}$, whose objects are smooth algebras and whose arrows are algebra homomorphisms. We show that these two categories are dual to one another. This result will be of crucial importance in our discussion of Einstein algebras and relativistic spacetimes.

There is a way to ``translate'' from the framework of smooth manifolds into the framework of smooth algebras. We call this translation $F$ and define it as follows.
\begin{itemize}
\item Given a smooth manifold $M$, $F(M) = C^{\infty}(M)$ is the algebra of smooth scalar functions on $M$.
\item Given a smooth map $\phi:M\rightarrow N$, $F(\phi)$ is the map $\hat{\phi}: C^{\infty}(N)\rightarrow C^{\infty}(M)$ given by $\hat{\phi}(f)=f\circ\phi$ for any $f\in C^{\infty}(N)$.
\end{itemize}

Before showing that it is a contravariant functor between $\textbf{SmoothMan}$ and $\textbf{SmoothAlg}$, we pause to make a couple of remarks about $F$. Let $M$ be a manifold and consider the algebra $F(M) = C^\infty(M)$ of smooth scalar functions on $M$. There is a natural correspondence between points in $M$ and elements of $|C^{\infty}(M)|$, given by the following map:
\begin{equation}
\theta_M:M\rightarrow|C^\infty(M)|\qquad \theta_M(p)(f)=f(p)
\label{thetadefinition}
\end{equation}
for any $p\in M$ and $f\in C^\infty(M)$. Note that $\theta_M(p)$ is indeed a homomorphism $C^\infty(M)\rightarrow\mathbb{R}$. One can easily verify that the algebra $C^\infty(M)$ is geometric, so we can consider the weak topology on $|C^\infty(M)|$. One then proves that relative to the weak topology the map $\theta_M:M\rightarrow |C^\infty(M)|$ is a homeomorphism \cite[7.4]{Nestruev}.

This fact allows one to prove the following simple result. The map $F$ translates a smooth manifold into a smooth algebra.

\begin{prop}
If $M$ is a smooth manifold, then $F(M) = C^\infty(M)$ is a smooth algebra \cite[7.5--7.6]{Nestruev}.
\label{smoothalgebra}
\end{prop}

The next important result about $F$ captures a sense in which the smooth maps between manifolds are characterized purely by their action on the algebras of smooth scalar functions.

\begin{prop}
Let $M$ and $N$ be smooth manifolds. A map $\phi:M\rightarrow N$ is smooth if and only if $\hat{\phi}(C^\infty(N))\subset C^\infty(M)$ \cite[7.16--7.18]{Nestruev}.
\label{smoothmaps}
\end{prop}

These results allow one to make precise a sense in which $F$ is indeed a ``translation'' from the framework of smooth manifolds into the framework of smooth algebras. We have the following result.

\begin{lem} \label{Fcontra}
$F:\textbf{SmoothMan}\rightarrow\textbf{SmoothAlg}$ is a contravariant functor.
\end{lem}

\begin{proof}
Proposition \ref{smoothalgebra} immediately implies that $F(M)$ is indeed an object in $\textbf{SmoothAlg}$ for every smooth manifold $M$. Let $\phi:M\rightarrow N$ be a smooth map. We need to show that the map $F(\phi)=\hat{\phi}: F(N)\rightarrow F(M)$ is an algebra homomorphism. Proposition \ref{smoothmaps} implies that $\hat{\phi}(f)\in F(M)$ for every $f\in F(N)$. One can easily verify that $\hat{\phi}$ preserves the vector space operations, the product, and the multiplicative identity, so $F(\phi)=\hat{\phi}$ is an algebra homomorphism. Furthermore, it is easy to see that $F$ preserves identities and reverses composition.
%if $\phi_1: M_1\rightarrow M_2$ and $\phi_2:M_2\rightarrow M_3$ are smooth maps then $F(\phi_2\circ\phi_1)=F(\phi_1)\circ F(\phi_2)$. And if $1_M:M\rightarrow M$ is the identity map on a smooth manifold $M$, then $F(1_M)=1_{F(M)}$.
This implies that $F:\textbf{SmoothMan}\rightarrow\textbf{SmoothAlg}$ is a contravariant functor.
\end{proof}

There is also a way to ``translate'' from the framework of smooth algebras into the framework of smooth manifolds. In order to describe this translation we need to do some work. Let $A$ be a smooth algebra. One can use the smooth algebraic structure of $A$ to define a smooth manifold $G(A)$. The underlying point set of the manifold $G(A)$ is the set $|A|$ of points of the algebra $A$.

Since $A$ is a smooth algebra, there is a covering of $|A|$ by open sets $\{U_k\}$ along with isomorphisms $i_k:A_{|U_k}\rightarrow C^{\infty}(\mathbb{R}^n)$ for some fixed $n$. We will use these open sets and isomorphisms to define charts $(U,\psi)$ on $|A|$. We first consider the maps
\[
h_k:A\rightarrow C^\infty(\mathbb{R}^n) \qquad h_k=i_k\circ \rho_{U_k},
\]
where $\rho_{U_k}: A\rightarrow A_{|U_k}$ is the restriction homomorphism. One can verify that the maps $|\rho_{U_k}|:|A_{|U_k}|\rightarrow U_k\subset|A|$ are homeomorphisms onto $U_k$ \cite[7.7-7.8]{Nestruev}. Since $i_k$ is an isomorphism and $|C^\infty(\mathbb{R}^n)|=\mathbb{R}^n$ \cite[3.16]{Nestruev}, it follows that $|h_k|=|\rho_{U_k}|\circ|i_k|$ is a homeomorphism $|h_k|:\mathbb{R}^n\rightarrow U_k$. We therefore define the charts $(U_k, \psi_k)$, where $\psi_k=|h_k|^{-1}$ for each $k\in\mathbb{N}$. One can verify that these charts are compatible \cite[7.10]{Nestruev}.

In addition to these charts $(U_k, \psi_k)$, we add charts of the form $(V\cap U_k, \psi_k)$ where $V\subset|A|$ is an open set and $k\in\mathbb{N}$. One can easily verify that these new charts are compatible both with each other and with the original charts $(U_k, \psi_k)$. Since the topology on $|A|$ is Hausdorff and since there is a countable cover of charts of the form $(U_k,\psi_k)$, if we throw in wholesale all of the charts on $|A|$ that are compatible with the charts of the form $(U_k, \psi_k)$ and $(V\cap U_k, \psi_k)$, then we will have defined a smooth (Hausdorff, paracompact) manifold \cite[Proposition 1.1.1]{MalamentGR}. We call this smooth manifold $G(A)$.

The smooth manifold $G(A)$ bears a close relationship to the original smooth algebra $A$. Indeed, the elements of $A$ correspond to smooth scalar functions on $G(A)$. This correspondence is given by the following map:
\begin{equation}
\eta_A:A\rightarrow FG(A)\qquad \eta_A:f\longmapsto (p\mapsto p(f))
\label{etadefinition}
\end{equation}
for every $f\in A$ and $p\in G(A)=|A|$. One can prove that for every $f\in A$ the function $p\mapsto p(f)$ is a smooth scalar function on $G(A)$, and furthermore, that the map $\eta_A$ is a bijection \cite[7.11]{Nestruev}. We can therefore think of elements of $A$ as smooth scalar functions on the manifold $G(A)$.

The translation $G$ from the framework of smooth algebras into the framework of smooth manifolds is defined as follows.
\begin{itemize}
\item Given a smooth algebra $A$, $G(A)$ is the smooth manifold defined above.
\item Given an algebra homomorphism $\psi :A\rightarrow B$, $G(\psi)$ is the map $|\psi|:|B|\rightarrow |A|$ between the manifolds $G(B)$ and $G(A)$.
\end{itemize}
Note that the definition of $G(\psi)$ makes sense since $G(B)$ and $G(A)$ have underlying point sets $|A|$ and $|B|$, respectively. Like Lemma \ref{Fcontra}, the following result captures a sense in which $G$ is a translation between these two frameworks.

\begin{lem} \label{Gcontra}
$G:\textbf{SmoothAlg}\rightarrow\textbf{SmoothMan}$ is a contravariant functor.
\end{lem}

\begin{proof}
We have already shown that $G(A)$ is a smooth manifold for every smooth algebra $A$. Let $\psi:A\rightarrow B$ be an algebra homomorphism. We need to show that $G(\psi)=|\psi|:|B|\rightarrow |A|$ is a smooth map between the manifolds $G(B)$ and $G(A)$. We begin by showing that
\begin{equation}
\widehat{|\psi|} \circ\eta_A=\eta_B\circ \psi
\label{etanaturality}
\end{equation}
For every $f\in A$ and $p\in |B|$ we see that following equations hold:
\begin{align*}
(\widehat{|\psi|}\circ\eta_A(f))(p)&=(\eta_A(f)\circ |\psi|)(p)\\
&=\eta_A(f)(p\circ \psi)\\
&=(p\circ\psi)(f)\\
&=(\eta_B\circ \psi(f))(p)
\end{align*}
The first equality follows from the definition of $\widehat{|\psi|}$, the second from the definition of $|\psi|$, the third from the definition of $\eta_A$, and the fourth from the definition of $\eta_B$. This establishes equation \eqref{etanaturality}. Since the maps $\eta_A$ and $\eta_B$ are bijections, \eqref{etanaturality} implies that $\widehat{|\psi|}=\eta_B\circ\psi\circ\eta_A^{-1}$. And this means that $\widehat{|\psi|}(FG(A))\subset FG(B)$. Proposition \ref{smoothmaps} then guarantees that $|\psi|:G(B)\rightarrow G(A)$ is a smooth map. One easily verifies that $G$ preserves identities and reverses composition,
%if $\psi_1:A_1\rightarrow A_2$ and $\psi_2:A_2\rightarrow A_3$ are homomorphisms then $G(\psi_2\circ\psi_1)=G(\psi_1)\circ G(\psi_2)$. And furthermore, if $1_A:A\rightarrow A$ is the identity homomorphism, then $G(1_A)=1_{G(A)}$.
so $G:\textbf{SmoothAlg}\rightarrow\textbf{SmoothMan}$ is a contravariant functor.
\end{proof}

We have shown that the maps $F:\textbf{SmoothMan}\rightarrow\textbf{SmoothAlg}$ and $G:\textbf{SmoothAlg}\rightarrow\textbf{SmoothMan}$ are contravariant functors. We now show that they are ``up to isomorphism'' inverses of one another. The following theorem makes this idea precise.

\begin{thm}
\label{man alg duality}
The categories $\textbf{SmoothMan}$ and $\textbf{SmoothAlg}$ are dual.
\end{thm}

\begin{proof}
We show that the families of maps $\eta:1_\textbf{SmoothAlg}\Rightarrow F\circ G$ and $\theta:1_\textbf{SmoothMan}\Rightarrow G\circ F$ defined in equations \eqref{thetadefinition} and \eqref{etadefinition} are natural isomorphisms. Since $F$ and $G$ are contravariant functors, this will imply that \textbf{SmoothMan} and \textbf{SmoothAlg} are dual categories.

We first consider $\eta$. We need to verify that for every smooth algebra $A$ the component $\eta_A:A\rightarrow FG(A)$ is an algebra isomorphism. We have already seen that $\eta_A$ is bijective. One easily checks that $\eta_A$ preserves the vector space operations, the product, and the multiplicative identity. Equation \eqref{etanaturality} implies that naturality square for $\eta$ commutes, so $\eta:1_\textbf{SmoothAlg}\Rightarrow F\circ G$ is a natural isomorphism.

We now consider $\theta$. We need to verify that for every smooth manifold $M$ the component $\theta_M:M\rightarrow GF(M)$ is a diffeomorphism. We already know that it is bijective. We show that $\hat{\theta}_M(FGF(M))\subset F(M)$ and then use Proposition \ref{smoothmaps} to conclude that $\theta_M$ is smooth. Let $f\in FGF(M)$. Since $\eta_{F(M)}:F(M)\rightarrow FGF(M)$ is a bijection there is some $g\in F(M)$ such that $\eta_{F(M)}(g)=f$. We then see that the following equalities hold for every $p\in M$:
\begin{align*}
\hat{\theta}_M(f)(p)= \hat{\theta}_M (\eta_{F(M)}(g))(p)=(\eta_{F(M)}(g)\circ\theta_M)(p)=\theta_M(p)(g)=g(p)
\end{align*}
The first equality holds by our choice of the function $g\in F(M)$, the second by the definition of $\hat{\theta}_M$, the third by the definition of $\eta_{F(M)}$, and the fourth by the definition of $\theta_M$. This implies that $\hat{\theta}_M(f)=g\in F(M)$. So we have shown that $\hat{\theta}_M(FGF(M))\subset F(M)$, which by Proposition \ref{smoothmaps} means that $\theta_M:M\rightarrow GF(M)$ is a smooth map. One argues in an analogous manner to show that $\theta_M^{-1}$ is smooth. Therefore $\theta_M:M\rightarrow GF(M)$ is a diffeomorphism.

We also need to show that the naturality square for $\theta$ commutes. Let $\phi:M\rightarrow N$ be a smooth map. We show that $GF(\phi)\circ\theta_M=\theta_N\circ\phi$. For every $p\in M$ and $f\in F(N)$ we see that the following equalities hold:
\begin{align*}
(|\hat{\phi}|\circ\theta_M(p))(f)&=(\theta_M(p)\circ\hat{\phi})(f)\\
&=\theta_M(p)(f\circ\phi)\\
&=f\circ\phi(p)\\
&=(\theta_N\circ\phi(p))(f)
\end{align*}
The first equality follows from the definition of $|\hat{\phi}|$, the second from the definition of $\hat{\phi}$, the third from the definition of $\theta_M$, and the fourth from the definition of $\theta_N$. Since $p$ and $f$ were arbitrary, this implies that $GF(\phi)\circ\theta_M=\theta_N\circ\phi$, so $\theta:1_\textbf{SmoothMan}\Rightarrow G\circ F$ is a natural isomorphism, and the categories $\textbf{SmoothMan}$ and $\textbf{SmoothAlg}$ are dual.
\end{proof}

Theorem \ref{man alg duality} allows one to identify the smooth algebra $A$ with $FG(A)=C^\infty(|A|)$ and the smooth manifold $M$ with $GF(M)=|C^\infty(M)|$. In addition, one can identify an algebra homomorphism $\psi:A\rightarrow B$ between smooth algebras with $FG(\psi)=\widehat{|\psi|}$ and a smooth map $\phi:M\rightarrow N$ between smooth manifolds with $GF(\phi)=|\hat{\phi}|$. We will implicitly appeal to these identifications in what follows.

\section{Einstein algebras and relativistic spacetimes}

The theory of Einstein algebras proceeds by taking a 4-dimensional smooth algebra $A$---which by Theorem \ref{man alg duality} corresponds to some smooth 4-dimensional manifold---and defining additional structure on it. This structure corresponds to the various fields that one requires to formulate general relativity. We begin by describing this structure.

%Let $A$ be a smooth geometric algebra.  A map $\xi: A \rightarrow \mathbb{R}$ is called a \emph{derivation} at a point $x \in |A|$ if it satisfies:
%\begin{enumerate}
%\item[(i)] linearity, i.e.
%$$\xi \left(\sum_{i = 1}^{k}\lambda_i f_i \right) = \sum_{i = 1}^{k}\lambda_i \xi(f_i ) \text{ for all } \lambda_i \in \mathbb{R}, \text{ } f_i \in A;  $$
%\item[(ii)] the Leibniz rule at $x$, i.e.,
%$$\xi(fg) = f(x)\xi(f) + g(x)\xi(g) \text{ for all } f, g \in A.$$
%\end{enumerate}
%The set of all derivations at $x$ is called the \emph{tangent space} at $x$, written $T_xA$.\\
%A smooth linear operator $\hat{X}: A \rightarrow A$ satisfying the Leibniz rule is called a \emph{(smooth) vector field} on $A$.

Let $A$ be a smooth algebra.  A \textit{derivation on $A$} is an $\mathbb{R}$-linear map $\hat{X}: A \rightarrow A$ that satisfies the Leibniz rule, in the sense that
$$
\hat{X}(fg) = f\hat{X}(g) + g\hat{X}(f)
$$
for all $f,g \in A$. The space of all derivations on $A$ is a module over $A$. We will use the notation $\hat{\Gamma}(A)$ to denote this module and $\hat{\Gamma}^*(A)$ to denote the dual module. The elements of the dual module $\hat{\Gamma}^*(A)$ are just the $A$-linear maps $\hat{\Gamma}^*(A)\rightarrow A$.  Derivations on $A$ allow one to define an analog to ``tangent spaces'' on smooth algebras. Given a derivation $\hat{X}$ on $A$ and a point $p\in |A|$, one can consider the linear map $\hat{X}_p:A\rightarrow\mathbb{R}$ defined by $\hat{X}_p(f)=\hat{X}(f)(p)$. The \textit{tangent space to $A$} at a point $p\in |A|$ is the vector space $T_pA$ whose elements are maps $\hat{X}_p: A \rightarrow \mathbb{R}$. The cotangent space to $A$ at a point $p\in |A|$ is defined similarly.

%\begin{rem}\label{rem: vector fields}
Derivations $\hat{X}$ on the smooth algebra $A$ naturally correspond to ordinary smooth vector fields $X$ on the manifold $G(A)=|A|$. The correspondence is given by
\begin{equation}\label{vector field correspondence}
\hat{X}(f)(p) =%\hat{X}_p(f) = X_p(f)
X_p(f)
\end{equation}
where $f \in C^{\infty}(|A|) = A$ and $p \in |A|$. This correspondence plays a crucial role in the following results, so we take a moment here to unravel the idea behind it. Given a derivation $\hat{X}$ on $A$, equation \eqref{vector field correspondence} defines a vector field $X$ on the manifold $G(A)$. This vector field $X$ assigns the vector $X_p$ to the point $p\in G(A)$, where the vector $X_p$ is defined by its action $f\mapsto \hat{X}(f)(p)$ on smooth scalar functions $f\in A$ on the manifold $G(A)$. One uses the fact that $\hat{X}$ satisfies the Leibniz rule to show that $X_p$ is indeed a vector at the point $p\in |A|$. One also verifies that the vector field $X$ is smooth.

Conversely, given a vector field $X$ on the manifold $G(A)$, equation \eqref{vector field correspondence} defines the derivation $\hat{X}$ on $A$. The derivation $\hat{X}$ maps an element $f\in A$ to the element of $A$ defined by the scalar function $X(f)$ on the manifold $G(A)$. It follows immediately that $\hat{X}$ is linear and satisfies the Leibniz rule. One can argue in a perfectly analogous manner to show that elements of $\hat{\Gamma}^*(A)$ correspond to smooth covariant vector fields on the manifold $G(A)$. Note also that given a point $p\in|A|$ the correspondence \eqref{vector field correspondence} allows one to naturally identify the elements $\hat{X}_p$ of the tangent space $T_pA$ and the vectors $X_p$ at the point $p$ in the manifold $G(A)$.
%\end{rem}

%HERE

%\begin{rem} \label{rem: vector fields}
%A vector field can be thought of as assigning to each $x \in |A|$ an element of its tangent space as follows. One can associate with any vector field $\hat{Y}$ a family of tangent vectors $\{\hat{Y} _x \in T_xA \}_{x \in |A|}$ where $\hat{Y} _x (f) = x(\hat{Y}(f))$ for all $x \in |A|$, $f \in A$.  The set of all (smooth) vector fields on $A$ is denoted $\hat{\Gamma}(A)$.  These form a module over $A$.
%\end{rem}

A \emph{metric} on a smooth algebra $A$ is a module isomorphism $\hat{g}: \hat{\Gamma}(A) \rightarrow \hat{\Gamma}^*(A)$ that is symmetric, in the sense that $\hat{g}(\hat{X})(\hat{Y})=\hat{g}(\hat{Y})(\hat{X})$ for all derivations $\hat{X}$ and $\hat{Y}$ on $A$. A metric $\hat{g}$ on $A$ induces a map $\hat{\Gamma}(A)\times\hat{\Gamma}(A)\rightarrow A$ defined by
$$
\hat{X}, \hat{Y}\longmapsto\hat{g}(\hat{X})(\hat{Y})
$$
Given a point $p\in|A|$, a metric on $A$ also induces a map $T_pA\times T_pA\rightarrow\mathbb{R}$ defined by $\hat{X}_p, \hat{Y}_p\mapsto \hat{g}(\hat{X},\hat{Y})(p)$.
We will occasionally abuse notation and use $\hat{g}$ to refer to all three of these maps, but it will always be clear from context which map is intended.

If $\hat{g}$ is a metric on an $n$-dimensional smooth algebra $A$ and $p$ is a point in $|A|$, then there exists an $m$ with $0\leq m \leq n$ and a basis $\xi_1, \ldots, \xi_n$ for the tangent space $T_p A$ such that
$$
\begin{array}{ccrl}
\hat{g}(\xi_i , \xi_i ) & = & +1 &\qquad \text{if }\quad 1\leq i\leq m\\
\hat{g}(\xi_j , \xi_j ) & = & -1 &\qquad \text{if }\quad m < j\leq n\\
\hat{g}(\xi_i , \xi_j ) & = & 0 &\qquad \text{if }\quad i\neq j
\end{array}
$$
We call the pair $(m, n-m)$ the \textit{signature} of $\hat{g}$ at the point $p\in|A|$. A metric $\hat{g}$ on $|A|$ that has signature $(1, n-1)$ at every point $p\in |A|$ is called a metric of \textit{Lorentz signature}.

We now have the resources necessary to begin discussing the theory of Einstein algebras. An \textit{Einstein algebra} is a pair $(A, \hat{g})$, where $A$ is a smooth algebra and $\hat{g}$ is a metric on $A$ of Lorentz signature. Before proving that the category of Einstein algebras is dual to the category of relativistic spacetimes, we need some basic facts about the relationship between metrics on algebras and metrics on manifolds.

\begin{lem} \label{metrics}
Let $M$ be an $n$-dimensional smooth manifold and let $A$ be an $n$-dimensional smooth algebra. Then the following all hold:
\begin{enumerate}
\item[(1)] If $g$ is a Lorentzian metric on $M$, then $\hat{g}$ is a Lorentzian metric on the algebra $F(M)=C^{\infty}(M)$, where $\hat{g}(\hat{X})(\hat{Y}):=g(X, Y)$;
\item[(2)] If $\hat{g}$ is a Lorentzian metric on $A$, then $|\hat{g}|$ is a Lorentzian metric on the manifold $G(A)=|A|$, where $|\hat{g}|(X, Y) := \hat{g}(\hat{X})(\hat{Y})$;
\item[(3)] If $g$ is a metric on a $M$, then $|\hat{g}| = g$;
\item[(4)] If $\hat{g}$ is a metric on $A$, then $\widehat{|\hat{g}|} = \hat{g}$.
\end{enumerate}
\end{lem}

\begin{proof}
%By Theorem~\ref{man alg duality}, there is a unique (up to smooth algebra isomorphism) algebra $A = C^{\infty}(M)$, and $M$ is the unique (up to diffeomorphism) manifold such that $|A| = M$.
%Let $\Gamma(M)$ be the  $g: \Gamma(M) \times \Gamma(M) \rightarrow C^{\infty}(M)$ be a Lorentzian metric on $M$, and let $X$ and $Y$ be vector fields on $M$. By remark~\ref{rem: vector fields}, the transformation $X \leftrightarrow \hat{X}$ is a well-defined correspondence between smooth vector fields on $M = |A|$ and derivations on $A = C^{\infty}(M)$. In particular, the map $\Phi: \Gamma(M) \rightarrow \hat{\Gamma}(A)$ as $X \mapsto \hat{X}$ is a module isomorphism, and so is the map $\Phi^*: \Gamma^*(M) \rightarrow \hat{\Gamma}^*(A)$ as $X^* \mapsto \hat{X}^*$.
%Thus a metric $g$ on $M$ expressed as a musical isomorphism $g: \Gamma(M) \rightarrow \Gamma^*(M)$ defines a symmetric module isomorphism $\hat{g} = \Phi^* \circ g \circ \Phi^{-1} : \hat{\Gamma}(A) \rightarrow \hat{\Gamma}^*(A)$, and a symmetric module isomorphism $\hat{g}$ defines a musical isomorphism $g =
%(\Phi^*)^{-1} \circ \hat{g} \circ \Phi$.

Let $g$ be a Lorentzian metric on $M$. One can easily verify that the map $\hat{g}:\hat{\Gamma}{(F(M))}\rightarrow\hat{\Gamma}^*(F(M))$ defined by $\hat{g}(\hat{X})(\hat{Y})=g(X,Y)$ is a symmetric module isomorphism, and therefore a metric on the smooth algebra $F(M)=C^\infty(M)$. It immediately follows from the bilinearity of $g$ that $\hat{g}$ is a module homomorphism; that $\hat{g}$ is bijective and symmetric follows from the fact that $g$ is non-degenerate and symmetric.

We also need to show that $\hat{g}$ has Lorentz signature. Let $p\in M$ and let $\xi_1, \ldots,\xi_n$ be an orthonormal basis (relative to the metric $g$) for the tangent space $T_pM$. Vectors at $p\in M$ can be naturally identified via \eqref{vector field correspondence} with elements of the tangent space $T_pF(M)$ to the algebra $F(M)=C^\infty(M)$. This identification and the definition of $\hat{g}$ immediately imply that $\hat{g}$ must have the same signature as $g$. So $\hat{g}$ is a metric of Lorentz signature on $F(M)$ and therefore (1) holds.
%
%such that
%
%$$
%   g(\xi_i , \xi_i )  = \left\{
%     \begin{array}{lr}
%       1 & \text{ if } i = 1 \\
%       -1 &  \text{ if } i > 1
%     \end{array}
%   \right.
%   $$
%
%   Note that $\eta: C^{\infty}(M) = A \rightarrow \mathbb{R}$ is a manifold derivation on $p \in M = |A|$ iff it is an algebraic derivation at $p$. Thus $\xi_1$,..., $\xi_{n+1}$ also forms an ON basis for $T_pA$.
%
One argues in an analogous manner to demonstrate (2).
%to show that a Lorentzian metric $\hat{g}$ on $A$ induces a Lorentzian metric $|\hat{g}|$ on the manifold $G(A)=|A|$.

If $g$ is a metric on $M$ and $X$ and $Y$ are vector fields on $M$, then $|\hat{g}|(X, Y) = \hat{g}(\hat{X})(\hat{Y}) = g(X, Y)$.
Furthermore, if $\hat{g}$ is a metric on $A$ and $\hat{X}$ and $\hat{Y}$ are derivations on $A$, then $\widehat{|\hat{g}|}(\hat{X})(\hat{Y}) = |\hat{g}|(X, Y) = \hat{g}(\hat{X})(\hat{Y})$. This immediately implies (3) and (4).
\end{proof}

Lemma \ref{metrics} captures a sense in which metrics on manifolds and metrics on smooth algebras encode exactly the same information. Each kind of structure naturally induces the other. This lemma strongly suggests, therefore, that we will be able to recover a sense in which general relativity and the theory of Einstein algebras are equivalent theories. Recovering this sense will require us to define a category of models for the theory of Einstein algebras. In order to do this, we need to discuss the ``structure-preserving maps'' between Einstein algebras.

Let $A$ and $B$ be smooth algebras with $\psi: A \rightarrow B$ an algebra homomorphism. Let $q$ be a point in $|B|$ and let $\hat{X}_q\in T_qB$ be an element of the tangent space to $B$ at $q$. The \textit{pullback} of $\hat{X}_q$ along the homomorphism $\psi$ is the element $\psi^*(\hat{X}_q)$ of $T_{|\psi|(q)} A$ defined by its action $\psi^*(\hat{X}_q)(f) = \hat{X}_q(\psi(f))$ on arbitrary elements $f\in A$. One again uses the correspondence \eqref{vector field correspondence} between vectors at the point $|\psi|(q)$ in the manifold $G(A)$ and elements of $T_{|\psi|(q)} A$ to verify that indeed $\psi^*(\hat{X}_q)\in T_{|\psi|(q)} A$. The pullback also allows us to use a homomorphism between smooth algebras to transfer other structures between the algebras. In particular, if $\hat{g}$ is a metric on $A$, the \textit{pushforward} $\psi_*(\hat{g})$ of $\hat{g}$ to $B$ is the map $\hat{g}:\hat{\Gamma}(B) \times \hat{\Gamma}(B) \rightarrow B$ defined by
$$
\psi_*(\hat{g})(\hat{X}, \hat{Y})(p) = \hat{g}(\psi^* (\hat{X}_p),\psi^* (\hat{Y}_p))
$$
for derivations $\hat{X}$ and $\hat{Y}$ on $B$. We now have the machinery to define the structure-preserving maps between Einstein algebras. If $(A, \hat{g})$ and $(B, \hat{g}')$ are Einstein algebras, an algebra homomorphism $\psi: A\rightarrow B$ is an \textit{Einstein algebra homomorphism} if it satisfies $\psi_*(\hat{g}) = \hat{g}'$. Einstein algebra homomorphisms are required to preserve both algebraic structure and the metric structure on the algebras.

We can now define the category of models $\textbf{EA}$ for the theory of Einstein algebras. The objects of the category $\textbf{EA}$ are Einstein algebras $(A, \hat{g})$, and the arrows are Einstein algebra homomorphisms. Our aim is to prove that $\textbf{EA}$ and $\textbf{GR}$ are dual categories. We first isolate two facts about the relationship between algebra homomorphisms and smooth maps in the following lemma.
\begin{lem} \label{lem:pushpull}
Let $\phi:M\rightarrow N$ be a smooth map between manifolds $M$ and $N$, and $\psi:A\rightarrow B$ be an algebra homomorphism between smooth algebras $A$ and $B$. Then the following both hold:
\begin{enumerate}
\item[(1)] $\widehat{\varphi_*(X_p)} = \hat{\varphi}^*(\hat{X}_p)$ for every vector $X_p$ at the point $p\in M$;
\item[(2)] $\psi^*(\hat{X}_q) = \widehat{|\psi|_*(X_q)}$ for every $\hat{X}_q \in T_qB$.
\end{enumerate}
\end{lem}

\begin{proof}
%By Theorem~\ref{man alg duality} $\psi: A \rightarrow B$ is an algebra homomorphism iff $|\psi|: |B| \rightarrow |A|$ is a smooth map between manifolds,  $A$ and $|A|$ correspond to one another uniquely up to isomorphism. Moreover, any smooth map $\varphi: M \rightarrow N$ between manifolds is such that $\hat{\varphi} : C^{\infty}(N) \rightarrow C^{\infty}(M)$ is an algebra isomorphism, and $\hat{\varphi} = |\hat{\varphi}|$. It thus suffices to prove $(1)$, as $(2)$ follows.
Let $X_p$ be a vector at $p$ in the manifold $M$ and $f\in C^\infty(N)=F(N)$. We demonstrate that (1) holds simply by computing the following.
$$
\widehat{\varphi_*(X_p)}(f) =  \varphi_*(X_p)(f)
= X_p(f \circ \phi)
= X_p(\hat{\varphi}(f))
= \hat{X}_p(\hat{\varphi}(f)) = \hat{\varphi}^*(\hat{X}_p)(f)
$$
The first and fourth equalities follow from the correspondence \eqref{vector field correspondence}, the second equality from the standard geometric definition of the pushforward $\phi_*$, the third from the definition of the map $\hat{\phi}$, and the fifth from the algebraic definition of the pullback $\hat{\phi}^*$.

The argument for (2) is perfectly analogous. Let $q\in|B|$ be a point with $\hat{X}_q\in T_q B$ and $f\in A$. We compute the following.
$$
\psi^*(\hat{X}_q)(f) = \hat{X}_q(\psi(f)) = X_q(\psi(f)) = X_q(f \circ |\psi|)
= |\psi|_*(X_q)(f) = \widehat{|\psi|_*(X_q)}(f)
$$
The first equality follows from the algebraic definition of the pullback $\psi^*$, the second and fifth follow from the correspondence \eqref{vector field correspondence}, the third by the definition of $|\psi|$, and the fourth by the standard geometric definition of the pushforward $|\psi|_*$.
\end{proof}

In conjunction with Theorem \ref{man alg duality}, Lemmas \ref{metrics} and \ref{lem:pushpull} allow us to define a pair of translations between the framework of Einstein algebras and the standard framework of general relativity. We first consider the natural way to translate relativistic spacetimes into Einstein algebras. We call this translation $J$ and define it as follows.
\begin{itemize}
\item Given a relativistic spacetime $(M, g)$, $J(M, g)=(C^\infty(M), \hat{g})$ is the Einstein algebra with underlying smooth algebra $C^\infty(M)$ and metric $\hat{g}$ defined in Lemma \ref{metrics}.
\item Given an isometry $\phi: (M, g)\rightarrow (M', g')$, $J(\phi)$ is the map $\hat{\phi}:C^\infty(M')\rightarrow C^\infty (M)$.
\end{itemize}
The translation $J$ is perfectly analogous to the contravariant functor $F$ described above. Indeed, as with $F$ we have the following simple result about $J$.

\begin{lem}
$J: \textbf{GR} \rightarrow \textbf{EA}$ is a contravariant functor.
\end{lem}

\begin{proof}
If $(M, g)$ is an object in $\textbf{GR}$, then it immediately follows that $J(M, g)$ is an object in $\textbf{EA}$. Proposition \ref{smoothalgebra} implies that $C^{\infty}(M)$ is a smooth algebra and Lemma \ref{metrics} implies that $\hat{g}$ is a metric of Lorentz signature on $C^\infty(M)$, so $J(M, g)$ is an Einstein algebra.

Now let $\phi: (M, g) \rightarrow (M', g')$ be an isometry. We need to show that $J(\phi)=\hat{\phi}:C^\infty(M')\rightarrow C^\infty(M)$ is an Einstein algebra homomorphism. Since $\phi$ is a smooth map,  Lemma \ref{Fcontra} guarantees that $\hat{\phi} : C^{\infty}(M') \rightarrow C^{\infty}(M)$ is an algebra homomorphism. It remains to show that $\hat{\varphi}_*(\hat{g}') = \hat{g}$. Let $\hat{X}$ and  $\hat{Y}$ be derivations on $C^\infty(M)$. We compute that
\begin{align*}
\hat{\varphi}_*(\hat{g}')(\hat{X}_p, \hat{Y}_p) &= \hat{g}'(\hat{\varphi}^*(\hat{X}_p), \hat{\varphi}^*(\hat{Y}_p))\\
&=\hat{g}'(\widehat{\varphi_*(X_p)}, \widehat{\varphi_*(Y_p)}) \\
&=g'(\varphi_*(X_p), \varphi_* (Y_p))\\
&=\varphi^*(g')(X_p, Y_p)= g(X_p, Y_p) = \hat{g}(\hat{X}_p, \hat{Y}_p)
\end{align*}
for every point $p\in M=|C^\infty(M)|$. The first equality follows from the definition of $\hat{\phi}_*$, the second from Lemma \ref{lem:pushpull}, the third from the definition of $\hat{g}'$, the fourth from the definition of $\phi^*$, the fifth since $\phi$ is an isometry, and the sixth from the definition of $\hat{g}$. This implies that $\hat{\phi}_*(\hat{g}')=\hat{g}$ and therefore that $J(\phi)=\hat{\phi}$ is an arrow $J(M', g')\rightarrow J(M, g)$. One easily verifies that $J$ preserves identities and reverses composition.
\end{proof}

There is also a way to translate from the framework of general relativity into the framework of Einstein algebras. We call this translation $K$ and define it as follows.
\begin{itemize}
\item Given an Einstein algebra $(A, \hat{g})$, $K(A, \hat{g})=(|A|, |\hat{g}|)$ is the relativistic spacetime with underlying manifold $|A|=G(A)$ and metric $|\hat{g}|$ defined in Lemma \ref{metrics}.
\item Given an Einstein algebra homomorphism $\psi:(A, \hat{g})\rightarrow (A', \hat{g}')$, $K(\psi)$ is the map $|\psi|:|A'| \rightarrow |A|$.
\end{itemize}
The translation $K$ is perfectly analogous to the contravariant functor $G$ described above. And again, we have the following result.

\begin{lem}
$K: \textbf{EA} \rightarrow \textbf{GR}$ is a contravariant functor.
\end{lem}

\begin{proof}
If $(A, \hat{g})$ is an object in $\textbf{EA}$, then it immediately follows that $K(A, \hat{g})$ is an object in $\textbf{GR}$. Indeed, we have already seen that $G(A)=|A|$ is a smooth manifold, and Lemma~\ref{metrics} implies that $|\hat{g}|$ is a metric on $|A|$, so $(|A|, |\hat{g}|)$ is a relativistic spacetime.

Now let $\psi: (A, \hat{g}) \rightarrow (A', \hat{g}')$ be an Einstein algebra homomorphism. Lemma \ref{Gcontra} guarantees that $K(\varphi) =  |\psi|: |A'| \rightarrow |A|$ is a smooth map. It remains to show that $|\psi|^*(|\hat{g}|) = |\hat{g}'|$. By Lemma \ref{metrics}, it will suffice to show that $|\psi|^*(g)=g'$. We let $X$ and $Y$ be vector fields on $|A'|$ and compute that
\begin{align*}
|\psi|^*(g)(X_p, Y_p) &= g(|\psi|_*(X_p), |\psi|_*(Y_p))\\
&= \hat{g}(\widehat{|\psi|_*(X_p)}, \widehat{|\psi|_*(Y_p)}) \\
&= \hat{g}(\psi^*(\hat{X}_p), \psi^*(\hat{Y}_p))\\
&= \psi_*(\hat{g})(\hat{X}_p, \hat{Y}_p)= \hat{g}'(\hat{X}_p, \hat{Y}_p)= g'(X_p, Y_p )
\end{align*}
for every point $p\in |A'|$. The first equality follows from the definition of $|\psi|^*$, the second from the definition of $\hat{g}$, the third from Lemma \ref{lem:pushpull}, the fourth from the definition of $\psi_*$, the fifth since $\psi$ is an Einstein algebra homomorphism, and the sixth from the definition of $\hat{g}'$. This implies that $|\psi|$ is an isometry and therefore an arrow $K(A', \hat{g}')\rightarrow K(A, \hat{g})$. One again easily verifies that $K$ preserves identities and reverses composition.
\end{proof}

We now have the resources necessary to prove our main result. The contravariant functors $J$ and $K$ realize a duality between the category of models for the theory of Einstein algebras and the category of models for general relativity. This result essentially follows as a corollary to Theorem \ref{man alg duality} along with parts (3) and (4) of Lemma \ref{metrics}.

\begin{thm}\label{thm:duality}
The categories $\textbf{EA}$ and $\textbf{GR}$ are dual.
\end{thm}

\begin{proof}
The proof exactly mirrors the proof of Theorem \ref{man alg duality}. We again show that the families of maps $\eta:1_{\textbf{EA}}\Rightarrow J\circ K$ and $\theta:1_{\textbf{GR}}\Rightarrow K\circ J$ defined in equations \eqref{thetadefinition} and \eqref{etadefinition} are natural isomorphisms. It follows from Theorem \ref{man alg duality} that the naturality squares  for $\eta$ and $\theta$ commute, so we need only check that the components of $\eta$ and $\theta$ are isomorphisms.

Let $(A, \hat{g})$ be an object in $\textbf{EA}$ and consider the component $\eta_{(A, \hat{g})}:A\rightarrow C^\infty(|A|)$. We have already seen in Theorem \ref{man alg duality} that $\eta_{(A, \hat{g})}$ is an algebra isomorphism. In addition, part (4) of Lemma \ref{metrics} implies that $\eta_{(A, \hat{g})}$ preserves the metric and therefore is an isomorphism between Einstein algebras. A perfectly analogous argument demonstrates that the components $\theta_{(M, g)}$ are isomorphisms between relativistic spacetimes.
%
%$\tilde{\eta}: 1_{\mathbf{EA}} \Rightarrow J \circ K$ to be a family of maps $\tilde{\eta}_{(A, \hat{g})} : (A, \hat{g}) \rightarrow JK(A, \hat{g})$ acting as $\tilde{\eta}_A$ on $A$ and  $\tilde{\eta}_{(A, \hat{g})}: \hat{g} \mapsto \widehat{|\hat{g}|}$ on the metric. That $\tilde{\eta}$ is a natural isomorphism follows from the fact that $\eta$ is a natural isomorphism (Theorem~\ref{man alg duality}), and that $JK(\hat{g}) = \widehat{|\hat{g}|} = \hat{g}$ by Lemma~\ref{metrics} (4).
%
%Define $\tilde{\theta}: 1_{\mathbf{GR}} \Rightarrow K \circ J$ to be a family of maps $\tilde{\theta}_{(M, g)} : (M, g) \rightarrow KJ(M, g)$ acting as $\tilde{\theta}_M$ on $M$ and  $\tilde{\theta}_{(M, g)}: g \mapsto |\hat{g}|$ on the metric. That $\tilde{\theta}$ is a natural isomorphism follows from the fact that $\theta$ is a natural isomorphism (Theorem~\ref{man alg duality}), and that $KJ(g) = |\hat{g}| = g$ by Lemma~\ref{metrics} (3).
%
%
\end{proof}

\citet{GerochEA} goes on to define other structures, analogous to, for instance, tensor fields and covariant derivative operators, in purely algebraic terms, using strategies similar to those used here to define derivations and metrics.  With this machinery, he argues, one can express any equation one likes, including Einstein's equation and various matter field equations, in algebraic terms.  In this way, one may proceed to do relativity theory using Einstein algebras and structures defined on them, in much the same way that one would using Lorentzian manifolds.  The functors $J$ and $K$, meanwhile, along with the results proved and methods developed here, provide a way of translating between equations relating tensor fields on a Lorentzian manifold $(M,g)$ and the corresponding structures defined on the Einstein algebra $J(M,g)$.  Moreover, we have a strong sense in which $J$ and $K$ preserve any possible empirical structure associated with general relativity, on either formulation, since any of the empirical content of general relativity on Lorentzian manifolds will be expressed using invariant geometrical structures such as curves, tensor fields, etc. or their algebraic analogues, and it is precisely this sort of structure that $J$ and $K$ preserve.

\section{Conclusion}

Theorem \ref{thm:duality} establishes a sense in which the Einstein algebra formalism is equivalent to the standard formalism for general relativity.  This sense of equivalence captures the idea that, on a natural standard of comparison, the two theories have precisely the same mathematical structure---and thus, we claim, the same capacities to represent physical situations.  Indeed, our proof of Theorem \ref{thm:duality} makes precise how any given model of one of the formulations may be transformed into a model of the other.  This transformation is ``loss-less'' in the sense that one can then recover the original model of the first formulation up to isomorphism. In particular, we do not find a compelling sense in which one of the formulations exhibits ``excess structure'' that the other eliminates.  Insofar as one wants to associate these two formalisms with ``substantivalist'' and ``relationist''---or at least, non-substantivalist---approaches to spacetime, it seems that we have a kind of equivalence between different metaphysical views about spatiotemporal structure.

Of course, it remains open to the person who wants to give these formalisms a metaphysical significance to say that one of them is more fundamental than the other.  In a sense, this was Earman's original proposal: an Einstein algebra, he believed, would give an intrinsic characterization of what a family of isometric spacetimes had in common, and in this sense, was more fundamental than any particular representation of the algebra.  Theorem \ref{thm:duality}, then, captures a sense in which the converse is equally true: a given relativistic spacetime may be thought of as offering an intrinsic characterization of what a family of isometric Einstein algebras have in common.  For our part, we see no no reason to choose between these approaches, at least in the absence of new physics that shows how one bears a closer relationship to future theories.\footnote{That new physics of the relevant sort is not implausible: for instance, \citet{Heller+Sasin1,Heller+Sasin2} have argued that Einstein algebras admit a natural generalization that allows one to treat classes of singularities that (they claim) are poorly treated by more traditional geometric methods.  See \citet{bain2003} for a philosophical discussion of these arguments and their significance.}  Both encode precisely the same physical facts about the world, in somewhat different languages.  It seems far more philosophically interesting to recognize that the world may admit of such different, but equally good, descriptions than to argue about which approach is primary.

% ----------------------------------------------------------------
\bibliographystyle{apalike}
\bibliography{biblio/gaugeMaster,biblio/masterbib}
\end{document}